\documentclass[oneside,american]{amsart}
\usepackage[T1]{fontenc}
\usepackage[utf8]{inputenc}
\usepackage{amsthm}
\usepackage{amssymb}

\makeatletter
\numberwithin{equation}{section}
\numberwithin{figure}{section}

\usepackage{url}

\makeatother

\theoremstyle{plain}
\newtheorem{thm}{\protect\theoremname}
\usepackage{babel}
\providecommand{\theoremname}{Theorem}

\begin{document}
\title[Uniform efficiency of the Ulrich--Wood sampler]{Uniform efficiency of the Ulrich--Wood sampler for the von Mises--Fisher
distribution}
\author{Sam Power}
\begin{abstract}
The standard approach to stochastic simulation from the von Mises--Fisher
distribution is a rejection sampler proposed by Ulrich and Wood. This
note provides a theoretical justification of its efficiency, lower-bounding
its acceptance probability uniformly over the dimension and concentration
parameter. A novel interpretation of the proposal is given which provides
some intuition for this efficiency.
\end{abstract}

\maketitle

\section{The von Mises--Fisher sampling problem}

With $n\geq1$, let $X\in\mathbb{S}^{n}\subset\mathbb{R}^{n+1}$,
and let $\mu\in\mathbb{S}^{n}$. The von Mises--Fisher distribution
with concentration parameter $\kappa\geq0$ and location parameter
$\mu$ has density 
\[
\pi^{\mathbb{S}}_{\kappa}\left(x\right)=\frac{\exp\left(\kappa\mu^{\top}x\right)}{Z_{n}\left(\kappa\right)}
\]
relative to the uniform probability measure on $\mathbb{S}^{n}$,
where
\[
Z_{n}\left(\kappa\right)=\mathbf{E}_{\mathrm{Unif}\left(\mathbb{S}^{n}\right)}\left[\exp\left(\kappa\mu^{\top}X\right)\right].
\]
 By rotational symmetry it suffices to sample $T=\mu^{\top}X\in\left[-1,1\right]$.
Under the uniform distribution on $\mathbb{S}^{n}$, the density of
$T$ on $\left[-1,1\right]$ is given by 
\[
h^{\mathbb{U}}_{n}\left(t\right)=c_{n}\left(1-t^{2}\right)^{n/2-1},\qquad c_{n}=\frac{\Gamma\left(\frac{n+1}{2}\right)}{\sqrt{\pi}\Gamma\left(\frac{n}{2}\right)}.
\]
Consequently, the density for $T$ under the von Mises--Fisher law
is 
\[
p^{\mathbb{U}}_{\kappa,n}\left(t\right)=\frac{\exp\left(\kappa t\right)h^{\mathbb{U}}_{n}\left(t\right)}{Z_{n}\left(\kappa\right)},\qquad Z_{n}\left(\kappa\right)=\int^{1}_{-1}\exp\left(\kappa t\right)h^{\mathbb{U}}_{n}\left(t\right)\,\mathrm{d}t.
\]
Once $T$ has been sampled, one may draw $W$ uniformly on the unit
sphere in $\mu^{\perp}$ and set $X=T\mu+\sqrt{1-T^{2}}W$. In this
regard, the substantive sampling problem is one-dimensional, as one
might expect.

A useful additional change of variables involves the hyperbolic transformation
$T=\tanh\Psi$, under which the uniform law on the sphere induces
the density 
\[
h^{\mathbb{R}}_{n}\left(\psi\right)=c_{n}\left(\cosh\psi\right)^{-n},\qquad\psi\in\mathbb{R},
\]
and the von Mises--Fisher law induces the density 
\[
p^{\mathbb{R}}_{\kappa,n}\left(\psi\right)=\frac{\exp\left(\kappa\tanh\psi\right)h^{\mathbb{R}}_{n}\left(\psi\right)}{Z_{n}\left(\kappa\right)}.
\]

\subsection{The Ulrich--Wood proposal}

We introduce here the rejection sampler of Ulrich and Wood \cite{ulrich1984computer,wood1994simulation},
using our hyperbolic transformation as the launching point. For $\psi_{0}\in\mathbb{R}$,
one can consider the use of the translated reference measure $h^{\mathbb{R}}_{n}$
as a proposal, i.e. the density $q^{\mathbb{R}}_{\psi_{0},n}\left(\psi\right)=c_{n}\left(\cosh\left(\psi-\psi_{0}\right)\right)^{-n}$.
This provides a valid rejection sampling envelope for any choice of
$\psi_{0}$. This proposal is straightforward to sample directly by
e.g., drawing $V\sim\mathsf{Beta}\left(\frac{n}{2},\frac{n}{2}\right)$,
setting $B=2V-1$ and then $\Psi=\psi_{0}+\tanh^{-1}B$.

The Ulrich--Wood choice can be obtained by seeking the minimal rejection
constant, solving the saddle-point problem (taking logarithms for
readability)
\[
\min_{\psi_{0}}\max_{\psi}\left\{ \kappa\tanh\psi-n\left(\log\cosh\psi-\log\cosh\left(\psi-\psi_{0}\right)\right)\right\} .
\]
Differentiation with respect to each argument gives the stationarity
conditions
\begin{align*}
\kappa\left(1-\tanh^{2}\psi\right)-n\left(\tanh\psi-\tanh\left(\psi-\psi_{0}\right)\right) & =0,\qquad & \tanh\left(\psi-\psi_{0}\right)=0,
\end{align*}
i.e. optimally, the proposal and target are tangent at $\psi=\psi_{0}$,
and the optimal $\psi_{0}$ satisfies $n\tanh\psi_{0}=\kappa\left(1-\tanh^{2}\psi_{0}\right)$,
i.e. $\psi_{0}=\frac{1}{2}\sinh^{-1}\left(\frac{2\kappa}{n}\right)$;
global optimality can be verified by routine arguments.

Under this setting, the acceptance test can be evaluated without the
normalising constant $Z_{n}\left(\kappa\right)$: given $\Psi$ so
drawn, accept if 
\[
\begin{aligned}\kappa\left(\tanh\psi_{0}-\tanh\Psi\right)-n\log\left(\frac{\cosh\psi_{0}\cosh\left(\Psi-\psi_{0}\right)}{\cosh\Psi}\right) & \leq E,\\
E & \sim\mathsf{Exponential}\left(1\right),
\end{aligned}
\]
noting that the expression on the left-hand side is always non-negative.

\subsection{Uniform efficiency}

Let $\alpha_{n}\left(\kappa\right)$ denote the acceptance probability
of the Ulrich--Wood rejection sampler.
\begin{thm}
For every $n\geq1$ and $\kappa\geq0$, the acceptance probability
satisfies
\[
\alpha_{n}\left(\kappa\right)\geq\sqrt{\frac{e}{2\pi}}=0.657744\ldots.
\]
More precisely,
\[
\inf_{\kappa\geq0}\alpha_{n}\left(\kappa\right)=L_{n}:=\frac{\Gamma\left(\frac{n+1}{2}\right)}{2\sqrt{\pi}}\left(\frac{2e}{n}\right)^{n/2},
\]
and the sequence $\left(L_{n}\right)_{n\geq1}$ is increasing. Consequently,
\[
\inf_{n\geq1}\inf_{\kappa\geq0}\alpha_{n}\left(\kappa\right)=L_{1}=\sqrt{\frac{e}{2\pi}}.
\]
\end{thm}

\begin{proof}
We first calculate the exact acceptance probability. The ratio of
the target and proposal densities satisfies
\[
\frac{p^{\mathbb{R}}_{\kappa,n}\left(\psi\right)}{q^{\mathbb{R}}_{\psi_{0},n}\left(\psi\right)}=\frac{\exp\left(\kappa\tanh\left(\psi\right)\right)}{Z_{n}\left(\kappa\right)}\left(\frac{\cosh\left(\psi-\psi_{0}\right)}{\cosh\left(\psi\right)}\right)^{n}\leq\frac{\exp\left(\kappa\tanh\left(\psi_{0}\right)\right)}{Z_{n}\left(\kappa\right)\left(\cosh\psi_{0}\right)^{n}}.
\]
and hence the marginal acceptance probability takes the form
\[
\alpha_{n}\left(\kappa\right)=Z_{n}\left(\kappa\right)\exp\left(-\kappa\tanh\psi_{0}\right)\left(\cosh\psi_{0}\right)^{n}.
\]
\\
\emph{Monotonicity in $\kappa$.} We compute the derivative of $\log\alpha_{n}\left(\kappa\right)$.
Standard results about exponential families give that $\left(\log Z_{n}\right)^{\prime}\left(\kappa\right)=\mathbf{E}^{\mathbb{U}}_{\kappa,n}\left[T\right]$.
Using the chain rule and applying the optimality condition for $\psi_{0}=\psi_{0}\left(\kappa\right)$
gives that 
\[
\frac{\mathrm{d}}{\mathrm{d}\kappa}\left\{ \kappa\tanh\psi_{0}-n\log\cosh\psi_{0}\right\} =\tanh\psi_{0},
\]
whereby $\left(\log\alpha_{n}\right)^{\prime}\left(\kappa\right)=\mathbf{E}^{\mathbb{U}}_{\kappa,n}\left[T\right]-\tanh\psi_{0}$.
Consider now the function $\iota_{\kappa,n}:t\mapsto nt-\kappa\left(1-t^{2}\right)$
for which $t=\tanh\psi_{0}$ is the unique root on $\left[-1,1\right]$
and which is positive to the right of this root. Explicit computations
(e.g., integration by parts and the Cauchy-Schwarz inequality) give
that 
\[
n\mathbf{E}^{\mathbb{U}}_{\kappa,n}\left[T\right]=\kappa\left\{ \mathbf{E}^{\mathbb{U}}_{\kappa,n}\left[1-T^{2}\right]\right\} \leq\kappa\left\{ 1-\mathbf{E}^{\mathbb{U}}_{\kappa,n}\left[T\right]^{2}\right\} ,
\]
i.e. $\iota_{\kappa,n}\left(\mathbf{E}^{\mathbb{U}}_{\kappa,n}\left[T\right]\right)\leq0=\iota_{\kappa,n}\left(\tanh\psi_{0}\right)$,
whereby $\mathbf{E}^{\mathbb{U}}_{\kappa,n}\left[T\right]\leq\tanh\psi_{0}$,
and hence $\left(\log\alpha_{n}\right)^{\prime}\leq0$, i.e. $\alpha_{n}$
is non-increasing in $\kappa$.\\
\emph{The fixed-$n$ worst case.} To evaluate the large-$\kappa$
limit, centre the hyperbolic coordinate by writing $\psi=\psi_{0}+\xi$,
and write 
\[
\alpha_{n}\left(\kappa\right)=c_{n}\int_{\mathbb{R}}\exp\left(\kappa\left(\tanh\left(\psi_{0}+\xi\right)-\tanh\psi_{0}\right)\right)\left(\frac{\cosh\left(\psi_{0}\right)}{\cosh\left(\psi_{0}+\xi\right)}\right)^{n}\,\mathrm{d}\xi.
\]
As $\kappa\to\infty$, one has $\psi_{0}\to\infty$ and, for every
fixed $\xi\in\mathbb{R}$, 
\begin{align*}
\kappa\left(\tanh\left(\psi_{0}+\xi\right)-\tanh\psi_{0}\right) & =\frac{n\sinh\psi_{0}\sinh\xi}{\cosh\left(\psi_{0}+\xi\right)}\longrightarrow\frac{n}{2}\left(1-\exp\left(-2\xi\right)\right),\\
\log\cosh\left(\psi_{0}+\xi\right)-\log\cosh\psi_{0} & \longrightarrow\xi.
\end{align*}
A routine dominated convergence argument allows us to send $\kappa\to\infty$,
and we therefore obtain that 
\begin{align*}
\lim_{\kappa\to\infty}\alpha_{n}\left(\kappa\right) & =c_{n}\int_{\mathbb{R}}\exp\left(\frac{n}{2}\left(1-\exp\left(-2\xi\right)\right)-n\xi\right)\,\mathrm{d}\xi\\
\left[z=\exp\left(-2\xi\right)\right]\qquad & =\frac{c_{n}\exp\left(\frac{n}{2}\right)}{2}\int^{\infty}_{0}z^{n/2-1}\exp\left(-\frac{n}{2}z\right)\,\mathrm{d}z\\
 & =\frac{\Gamma\left(\frac{n+1}{2}\right)}{2\sqrt{\pi}}\left(\frac{2e}{n}\right)^{n/2}=L_{n},
\end{align*}
whereby we conclude that $\inf_{\kappa\geq0}\alpha_{n}\left(\kappa\right)=L_{n}$.
\\
\emph{Monotonicity of the worst-case rates in $n$.} Extend $L_{n}$
to real $n\geq1$. The classical integral representation for the digamma
function (Equation 5.9.12 of \cite{NIST:DLMF}) gives that
\[
\left(\log L\right)^{\prime}\left(n\right)=\frac{1}{2}\int^{\infty}_{0}\exp\left(-n\mu\right)\left(\frac{1}{\mu}-\frac{1}{\sinh\mu}\right)\,\mathrm{d}\mu>0,
\]
whereby the monotonicity of $L$ follows. Separately, Stirling's formula
gives that $\lim_{n\to\infty}L_{n}=\frac{1}{\sqrt{2}}$.\\
\emph{The explicit global worst case.} At $n=1$, $L_{1}=\frac{\Gamma\left(1\right)}{2\sqrt{\pi}}\sqrt{2e}=\sqrt{\frac{e}{2\pi}}$.
Because $L_{n}$ is increasing, this is the global infimum over both
$n$ and $\kappa$, approached in the regime $n=1$ and $\kappa\to\infty$.
\end{proof}

\subsection{Comments and outlook}
\begin{enumerate}
\item The theorem implies the uniform expected-proposal bound 
\[
\sup_{n\geq1,\;\kappa\geq0}\mathbf{E}\left[N\right]=\sqrt{\frac{2\pi}{e}}=1.52035\ldots,
\]
and so the classical sampler does not deteriorate in any regime of
dimension and concentration.
\item The hyperbolic coordinate gives some insight as to the efficiency
of this particular proposal: the original density $\propto\left(\cosh\psi\right)^{-n}$
is tilted by the factor $\exp\left(\kappa\tanh\psi\right)$, tilting
its mass rightwards, and so the proposal $q^{\mathbb{R}}_{\psi_{0},n}$
translates its mass rightwards accordingly to follow the shift in
mode.
\end{enumerate}

\section{Acknowledgements}

This article arose out of a project proving uniform efficiency for
an alternative sampler of the von Mises--Fisher distribution, where
conversations with OpenAI ChatGPT were used to survey and compare
to previous work. Due to an initial mis-reading on my part, I had
\emph{not} expected the Ulrich--Wood sampler to be uniformly efficient,
and so I posed an off-hand question about it to ChatGPT in search
of an ostensible counter-example, which would have served as motivation
for the alternative sampler. 

Instead, the response quickly established that uniform efficiency
\emph{did} hold (with the explicit, sharp constant given here) and
indicated that this result was not in the literature already, which
I corroborated independently. I then set about re-writing the result
from scratch, incorporating also the hyperbolic interpretation of
the proposal (which was present in some of my other notes on directional
sampling problems).

\bibliographystyle{plain}
\bibliography{vmf}

@article{wood1994simulation,
	author = {Wood, Andrew TA},
	doi = {10.1080/03610919408813161},
	journal = {Communications in Statistics---Simulation and Computation},
	number = {1},
	pages = {157--164},
	publisher = {Taylor \& Francis},
	title = {Simulation of the {v}on {M}ises {F}isher distribution},
	volume = {23},
	year = {1994}}

@article{ulrich1984computer,
	author = {Ulrich, Gary},
	doi = {10.2307/2347441},
	journal = {Journal of the Royal Statistical Society: Series C (Applied Statistics)},
	number = {2},
	pages = {158--163},
	publisher = {Wiley Online Library},
	title = {Computer generation of distributions on the m-sphere},
	volume = {33},
	year = {1984}}

@misc{NIST:DLMF,
	howpublished = {\url{https://dlmf.nist.gov/}, Release 1.2.7 of 2026-06-15},
	key = {{\relax DLMF}},
	note = {F.~W.~J. Olver, A.~B. {Olde Daalhuis}, D.~W. Lozier, B.~I. Schneider, R.~F. Boisvert, C.~W. Clark, B.~R. Miller, B.~V. Saunders, H.~S. Cohl, and M.~A. McClain, eds.},
	title = {{\it NIST Digital Library of Mathematical Functions}},
	url = {https://dlmf.nist.gov/}}

\end{document}